\documentclass[11pt]{article}
\usepackage{amssymb,epsfig,url,pstricks,ted,wide}
\newtheorem{claim}{Claim}
\def\desd#1{\textsc{d#1}}
\begin{document}
\title{Separation of Circulating 
Tokens\thanks{Research supported by NSF Grant 0519907.}} 
\author{Kajari Ghosh Dastidar and Ted Herman \\
Department of Computer Science, University of Iowa}
\date{August 2009 (revised October 2009)}
\maketitle

\begin{abstract}
\noindent
Self-stabilizing distributed control is often modeled by 
token abstractions.  A system with a single token may 
implement mutual exclusion;  a system with multiple tokens
may ensure that immediate neighbors do not simultaneously
enjoy a privilege.  For a cyber-physical system, tokens
may represent physical objects whose movement is controlled.   
The problem studied in this paper is to 
ensure that a synchronous system with $m$ circulating 
tokens has at least $d$ distance between tokens.  
This problem is first considered in a ring where $d$ is given
whilst $m$ and the ring size $n$ are unknown.  The protocol 
solving this problem can be uniform, with all processes running
the same program, or it can be non-uniform, with some processes
acting only as token relays.  The protocol for this first 
problem is simple, and can be expressed with Petri net formalism.  
A second problem is to maximize $d$ when $m$ is given, and $n$ is
unknown.  For the second problem, the paper presents a 
non-uniform protocol with a single corrective process.  
\end{abstract} 
{\footnotesize
Keywords:  Self-stabilization, Petri Nets, Token Rings, Sensor Networks. \\
General Terms:  Algorithms, Reliability \\
Subjects: 	Distributed, Parallel, and Cluster Computing (cs.DC) \\
ACM classes: 	C.2.4 \emph{Distributed applications}; D.1.3 \emph{Distributed programming}; 
                D.2.2 \emph{Petri nets}; C.3 \emph{Real-time and embedded systems} \\
Report number:	TR09-02 Department of Computer Science, University of Iowa \\
Extended Abstract: \cite{HG09} 
}

\section{Introduction}
\label{section:introduction}

Distributed computing deals with the interaction of 
concurrent entities.  Asynchronous models permit irregular 
rates of computation whereas pure synchronous models can impose 
uniform steps across the system.  For either mode of concurrency
the application goals may benefit from controlled reduction of 
some activity.  Mutual exclusion aims to reduce the activity to one
process at any time;  some scheduling tasks require that 
certain related processes not be active at the same time.   
System activation of a controlled functionality 
is typically abstracted as a process \emph{having a token}, 
which constitutes permission to engage in some controlled action. 
Many mechanisms for regulating token
creation, destruction, and transfer have been published.  
This paper explores a mechanism based on timing information in 
a synchronous model.  In a nutshell, each process has one or 
more timers used to control how long a token rests or moves to
another process.  An emergent property of a protocol using this
mechanism should be that tokens move at each step, tokens visit
all processes, and no two tokens come closer than some given 
distance (or, alternatively, that tokens remain as far apart 
as possible).  The challenge, as with all self-stabilizing 
algorithms, is that tokens can initially be located arbitrarily
and the variables encoding timers or other variables may have 
unpredictable initial values.   

One motivating application is physical process control,
as formalized by Petri nets.  The tokens of a Petri net can 
represent physical objects.  As an example, one can imagine
a closed network where some objects are conveyed from place
to place, with some physical processing (loading, unloading,
modifications to parts) done at each place.  For the health
of the machinery it may be useful to keep the objects at some
distance apart, so that facilities at the different places have 
time to recharge resources between object visits.  Figure 
\ref{figure:petri} partially illustrates such a situation, 
with an unhealthy initial state (three objects are together
at one place).  The circuit of the moving objects is a ring
for this example.  The formalism of Petri nets allows us to 
add additional places, tokens and transitions so that a 
self-stabilizing network can be constructed:  eventually, the 
objects of interest will be kept apart by some desired 
distance.  Section \ref{section:known-sep} presents a self-stabilizing
algorithm for this network.  
 
\begin{figure}[ht]
\quad
\begin{minipage}{0.45\columnwidth}
{\footnotesize
The figure shows a large ring and two smaller rings, where each 
smaller ring is connected by a joint transition (which can only
fire when a token is present on each of its inputs) to the larger
ring.  On the right side a portion of the larger (clockwise) token ring
is represented, with three tokens shown resting together at the same
place.  Two other smaller, counterclockwise rings are partially 
shown on the left side, each with one token.  The joint transition 
will prevent the three resting tokens from firing until the 
token on the smaller ring completes its traversal.  Thus the 
smaller rings, each having exactly
one token in any state, behave as delay mechanisms.  The algorithm
given in Section \ref{section:known-sep} uses conventional process 
notation instead of a Petri net, and the smaller rings are replaced
by counters in a program.
}
\end{minipage}
\quad 
\framebox{
\begin{minipage}{0.45\columnwidth}
\epsfig{file=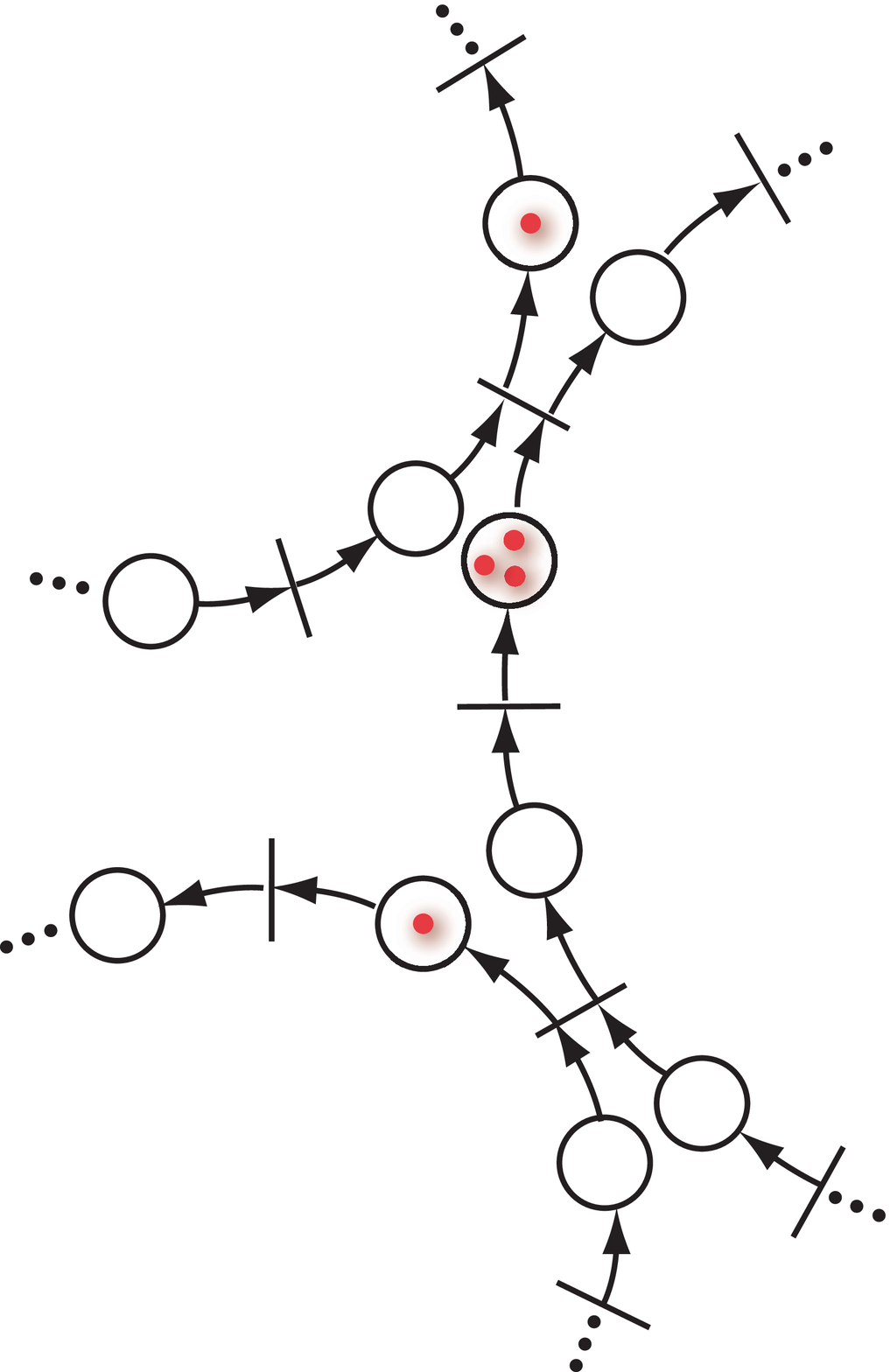,height=9cm,width=5.8cm}
\end{minipage}
}
\caption{Petri Net Embodiment}
\label{figure:petri}
\end{figure}

Another motivating application comes from wireless sensor networks, 
where power management is important.  A strategy for limiting power
consumption is to limit the number of sensors that are on at any 
time, presumably selecting enough sensors to be on for adequate 
coverage of a field of interest, yet rotating which nodes deploy 
sensors over time, to extend lifetime and to improve robustness 
with regard to variation in sensor calibration.  
One solution to this problem would be to use
clock synchronization, with a periodic schedule for sensing activity  
based on a global time.  Alternatively,   
token circulation could be considered to activate 
sensors.  Unlike a schedule purely based on synchronized clocks, 
a token-based solution provides some assurance and feedback 
in cases where nodes are faulty (\emph{e.g.,} when a token cannot be 
passed from one node to another due to a failure, such failure
may be recognized and an alarm could be triggered).  
The abstraction of tokens put into messages may also allow aggregated 
sensor data or commands to be carried 
with a token, further enabling application behavior.  Keeping tokens
apart may relate to coverage goals for the sensor network:  if tokens
circulate in parallel and satisfy some distance constraint between 
them, then the sensors that are on at any time may provide adequate 
spatial diversity over the field of interest.  Questions of 
satisfactory or optimal coverage of a field are  
beyond the scope of this paper.  Our investigation is 
confined to the problem of self-stabilizing circulation of tokens
with some desired separation between them.

\paragraph{Related Work.}

Perhaps the earliest source on self-stabilization is
\cite{Dij73}, which briefly presents an algorithm to distribute
$N$ points equally on a circle.  The algorithm given in Section
\ref{section:unknown-sep} distributes $m$ tokens equally around
a ring, however the objective is a behavior (circulating
tokens) rather than a final state.  Papers on coordinated
robot behavior, for example \cite{AFSY08,FPS08,CFPS03},
are similar to \cite{Dij73} in that a geometric, physical
domain is modeled.  Most such papers consider a final robot
configuration as the objective of distributed control and give
the robots powerful vision and mobility primitives.  Like the
example of robot coordination, our work can have a physical
control motivation, but we have a behavior as the objective.
For results in this paper, the computation model is discrete
and fully synchronous, where processes communicate only with
neighbors in a ring.  As for the sensor network motivation
sketched above, duty-cycle scheduling while satisfying coverage
has been implemented \cite{HKLY+06} (numerous network protocol
and system issues are involved in this task \cite{WX06}).
These sensor network duty-cycle scheduling efforts are not
self-stabilizing to our knowledge.

Within the literature of self-stabilization, a related problem
is model transformation.  If an algorithm $P$ is correct for
serial execution, but not for a parallel execution, then one
can implement a type of scheduler that only allows a process
$p$ of $P$ to take a step provided that no neighbor $q$ is
activated concurrently \cite{GH99};  this type of scheduling
is known to correctly emulate a serial order of execution.
The problem we consider, separating tokens by some desired
distance $d$, can be specialized to $d=2$ and be comparable to
such a model transformation.  For larger values of $d$, the
nearest related work is the general stabilizing philosopher
problem \cite{DNT07}, which considers conflict graphs between
non-neighboring philosophers.  By equating philosopher activity
(dining) to holding a token \cite{DNT07}, we get a solution to
the problem of ensuring tokens are at some desired distance,
and also allowing tokens to move as needed.  The synchronous
token behavior in this paper differs from the philosopher
problem because token circulation here is not demand-based; 
therefore solutions to the problem are obtained 
through the regular timing of token circulation.

Literature on self-stabilizing mutual exclusion includes token
abstractions \cite{Dol00}, generalizations of mutual exclusion
to $k$-exclusion or $k$-out-of-$\ell$ exclusion \cite{DHV03},
multitoken protocols for a ring \cite{FDS95}, and group mutual
exclusion \cite{CP00}.	Such literature does not constrain
tokens to be separated by some desired distance $d$ (unlike
the philosopher problem cited above), which differentiates
our work from previous multitoken protocols.  However, in
applying the methods of this paper to some applications, it
can be useful to employ self-stabilizing token or multitoken
protocols at a lower layer: Section \ref{section:appmodel}
expands on this idea of using self-stabilizing token protocols
as a basis for our work.

The idea of using the timing of token arrival to control
distributed behavior has previously investigated for balancing
(or counting) networks, which can be seen as abstractions for 
scheduling.  A token in a balancing network represents a 
locus of control;  the path of a token over time describes the 
history of accesses that one process performs on distinct shared memory
objects.  States of the junctions in these networks change as tokens
arrive and depart, and the state of a junction determines where an
arriving token will be next routed.  The relative timing of token 
arrivals to the network, and within the network, thus determine 
the pattern of flow through the network. Though such networks 
typically presume a properly initialized state,   the idea of a 
self-stabilizing behavior in a balancing network has been proposed 
\cite{HT03}.  Balancing networks are generally open networks, where
processes arrive, traverse the networks, and exit;  presumably, 
the output of such a network could be 
fed back into the same network to make a closed system. 
Balancing networks are chiefly intended for asynchronous execution, 
where the objective is to obtain some pattern in the history of arrivals 
of processes to selected shared objects.  Our goal is different: we suppose  
synchronous processes, with the time objective of keeping tokens
some distance $d$ apart at all times.  As an interesting aside, 
we note that algebraic (matrix) approaches have been found valuable both for Petri 
net analysis \cite{FCOQ92} and for combinatorial analysis of balancing networks 
\cite{CM96}. 

If we move beyond guaranteed behavior in discrete-time models 
of circulating tokens to stochastic behavior of moving particles 
in large networks, then statistical physics literature 
on traffic may be relevant.  Recent investigations 
consider capacity and efficiency metrics for flows of traffic \cite{Helbing01}, 
sometimes finding that separation between entities is important 
to shape traffic in the aggregate.  Experiments have shown that 
the density of vehicles on a unidirectional circle cannot exceed a 
critical threshold without traffic jams appearing \cite{SFK+08}; similar
results appear to hold for complex networks, validated by simulation 
\cite{Holme03,ZLPY05}.

\section{Desired Behavior}
\label{section:desiderata}

Desired properties of a token circulation protocol are  
labeled as \desd{1}--\desd{5} below.  
\begin{quote}
\desd{1}\quad At any time, $m$ tokens are present in the system. 

\desd{2}\quad The minimum distance between any two tokens is at least $d$.

\desd{3}\quad 
A token moves in each step from one process to a neighboring process.

\desd{4}\quad Every process has a token equally often; \emph{i.e.,} in an 
execution of $k$ steps, for any process $p_i$, there is a token at 
$p_i$ for approximately $k\cdot m/n$ steps.

\desd{5}\quad Following a transient failure that corrupts state variables 
of any number of processes, the system automatically recovers to 
behavior satisfying \desd{1}--\desd{4}.
\end{quote}
In many topologies, not all of \desd{1}--\desd{5} are achievable.  
As an instance, for \desd{3} to hold, the center node of a star topology 
or a simple linear chain is necessarily visited 
by tokens more often than other nodes, conflicting with \desd{4}.  
The constructions of this paper
are able to satisfy \desd{1}--\desd{5} for a ring topology.  Though it 
is straightforward to map a virtual ring on a complete walk over an 
arbitrary network, property \desd{2} may not hold:  
nodes at distance $d$ in a virtual ring could be at much smaller 
distance in the base network.  An example of a virtual ring is 
presented in Section \ref{section:conclusion}.

\section{Motivating Example} \label{section:motivation}

Though \desd{4} cannot be achieved for a star topology in which 
tokens circulate the network, there is a simple case where separation
of tokens can be obtained in an open network.  Figure \ref{figure:chain}
shows how distance between tokens can be enforced almost trivially,
by throttling the rate of tokens injected into the network.   

\begin{figure}[ht] 
\quad \framebox{ \begin{minipage}{0.17\columnwidth}
\epsfig{file=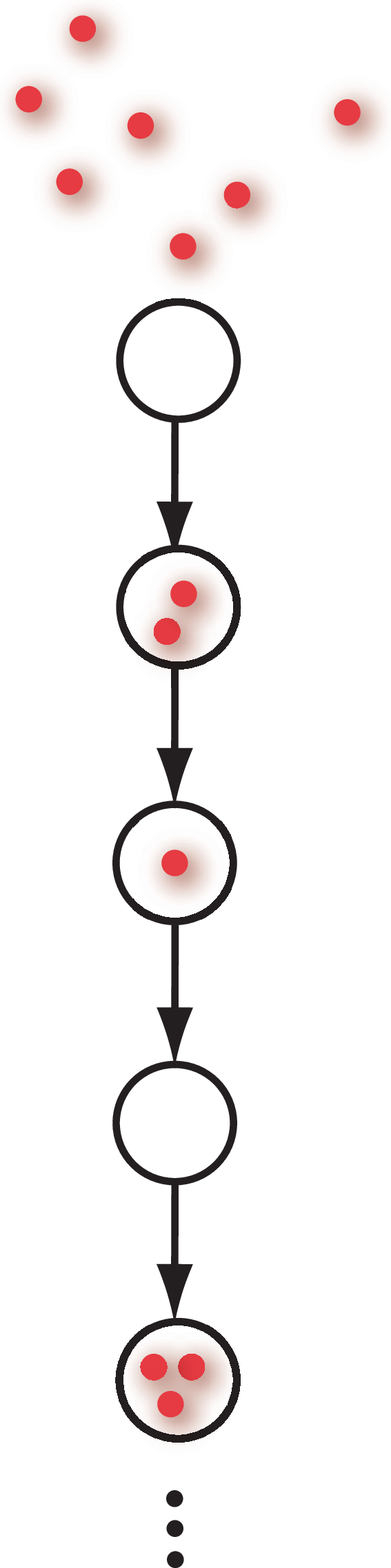,height=9cm,width=2.5cm} \end{minipage}}
\quad\framebox{ \begin{minipage}{0.2\columnwidth}
\epsfig{file=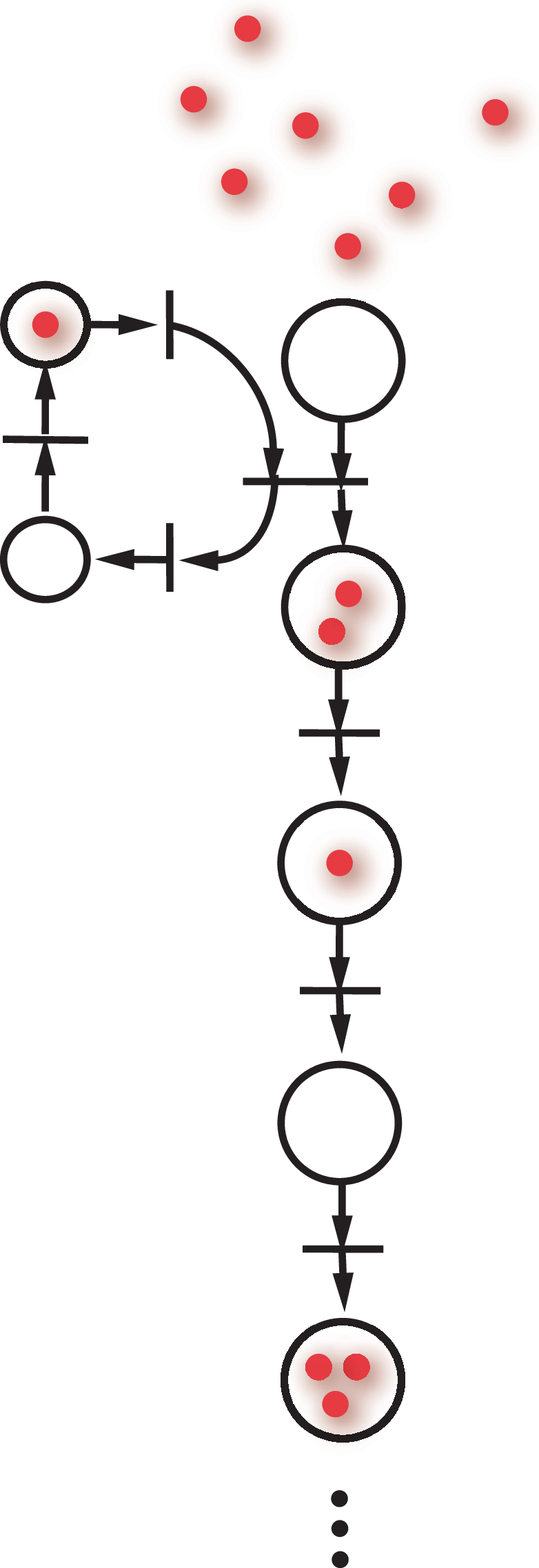,height=9cm,width=3cm} \end{minipage}}
\quad\begin{minipage}{0.5\columnwidth}
{\footnotesize
Illustrated on the far left is an open system consisting of 
a chain of processes, $p_1$, $p_2$, \ldots, with $p_1$ being
the topmost process illustrated.  Tokens are shown as dots, 
with a number of ``loose tokens'' above the chain representing
new tokens arriving from outside the system
to $p_1$.  Each process $p_i$ releases at most one 
token in each round to $p_{i+1}$.  The aim for this system is to 
ensure that, eventually, no two tokens are closer than some 
distance $d$ in the subchain from $p_2$ downward (we cannot prevent
the accumulation of tokens are $p_1$ in this open system).  On the
immediate left is a simple delay mechanism shown as a Petri net;  
the small subring and the joint transition between $p_1$ and $p_2$ 
ensures that tokens do not arrive in each round to $p_2$.  By 
adjusting the size of the subring, the target distance $d$ can 
be obtained.  }
\end{minipage}
\caption{Open System: Linear Chain} \label{figure:chain}
\end{figure}

The delay mechanism between $p_1$ and $p_2$ is shown as a token ring
conjoined to the chain.   We use simpler notation for this 
later in the paper:  a timer counts rounds between the times that
tokens are released.  If a new token arrives when the counter is zero, 
which is equivalent to a waiting token on the ring at the joint 
transition, then the new token is passed from $p_1$ to $p_2$ without
additional waiting.  The event of passing a token from $p_1$ restarts
the timer;  if a new token arrives with the counter is nonzero, then 
the new token will have to wait for the counter to reach zero before
it can be released.  This example reveals the strategy for separating 
tokens, namely to inject delay when needed.  A question arising from
Figure \ref{figure:chain} is whether the same, simple idea can work
in a closed system.  What happens, for example, if the output from
this chain, say at $p_n$, feeds back to $p_1$?  Will this simple 
delay suffice for stabilization?  Section \ref{section:known-sep} 
answers this question positively.

\section{Notation and Model}
\label{section:model}

Consider a ring of $n$ processes executing 
synchronously, in lock-step.  Each process perpetually
executes steps of a program, which are called  
\emph{local steps}.  In one \emph{global step}, every
process executes a local step.  Programs are structured
as infinite loops, where the body of a loop contains
statements that correspond to local steps.  We assume
that all processes execute the loop steps in a coordinated
manner:  for processes running the same program, all 
of them execute the first statement step in unison.  
Similarly, if two processes run distinct programs, we 
suppose they begin the body of the loop together, which
may entail padding the loop of one program to be the same
number of steps as the other program.  This assumption 
about coordination of steps is for convenience of presentation,
since it is possible to engineer all programs to have a loop
body with a single, more powerful step.  The execution of 
all steps in the loop, from first to last statement, is 
called a \emph{round}.     
\par
The notion of distance between locations in the ring can 
be measured in either clockwise or counterclockwise direction.
In program descriptions and proof arguments, it is convenient
to refer to the clockwise (counterclockwise) neighbor of a 
process using subscript notation:  process $p_i$'s clockwise
neighbor is $p_{i+1}$ and its counterclockwise neighbor is
$p_{i-1}$.  The distance from $p_i$ to itself is zero, the
clockwise distance from $p_i$ to $p_{i+1}$ is one, and the 
counterclockwise distance from $p_i$ to $p_{i+1}$ is 
$n-1$;  the counterclockwise distance from $p_i$ to $p_{i-1}$
is one, and general definitions of distance between $p_i$ and $p_j$ 
for arbitrary ring locations can be defined inductively.  
The counterclockwise neighbor of $p_i$ is called the 
\emph{predecessor} of $p_i$, and the clockwise neighbor is called
the \emph{successor}. 
\par
The local state of a process $p_i$ is specified by giving values
for its variables.  The global state of the system is an 
assignment of local states for all processes.  A protocol, specified
by giving programs for each $p_i$, should satisfy the 
desiderata of Section \ref{section:desiderata}.  A protocol is 
\emph{self-stabilizing} if, eventually, \desd{1}--\desd{5} hold 
throughout the suffix of any execution.  For simplicity, in the 
presentation of our protocols, we make some unusual model choices:
in one case, $p_i$ assigns to a variable of $p_{i-1}$;  and we 
assume that $m$ tokens are present in any initial state of any 
execution.  After presenting programs in Section \ref{section:known-sep}, 
we discuss in Subsection \ref{section:appmodel} these choices 
in reference to the two illustrative applications, 
Petri nets and sensor networks. 
 
\section{Protocol with Known Separation} 
\label{section:known-sep}

This section presents a protocol to achieve and maintain a separation
of at least $C+1$ links between tokens in the unidirectional ring.  
An implementation of the protocol uses four instantiation parameters, 
$n$, $m$, $C$, and the choice of which of two programs are used for 
nodes in the ring.  Only the separation parameter $C$ is used
in the protocol, as the domain of a counter, whereas the ring size $n$ and 
the number of tokens $m$ are unknown for the programs.   The 
separation by $C+1$ links cannot be realized for arbitrary $n>1$ and $m>1$;
we require that 
\begin{eqnarray} 
\label{eqn:separation}
m(C+1) \quad&\leq&\quad n
\end{eqnarray}
The protocol consists of two programs 
\textsf{delay} and \textsf{relay}.   
At least one process in the system executes
the \textsf{delay} and any processes not running \textsf{delay} 
run the \textsf{relay} program.  Processes running either program
have two variables, $q$ and $r$;  a process running
delay has an additional variable $c$.  
To specify the variable of a particular process, we subscript 
variables, for instance, $q_i$ is the $q$ variable of process $p_i$.   
The domains of $q$ and $r$ are nonnegative  
integers;  the domain of $c$ is the range of integers in $[0,C]$.    
\par
The $q$ and $r$ variables model the abstraction of tokens in a ring.
At any global state $\sigma$, process $p_i$ said to have $t$ tokens if 
$r_i+q_i=t$.  We say that $k$ tokens are \emph{resting} at $p_i$ if   
$r_i=k$, and $\ell$ tokens are \emph{queued} (for moving forward) if
$q_i=\ell$.  The objective of the protocol is to circulate  
$m$ tokens around the ring so that the distance from one token to
the next (clockwise) token exceeds parameter $C$, and in each 
round every token moves from its current location to the successor.  
In some cases, it is handy to refer to the value of 
a variable at a particular state in an 
execution.  The term $r_i^{\sigma}$ denotes the value of $r_i$ at a 
state $\sigma$.  In most cases, the state is implicitly the present (current)
state with respect to a description or a predicate definition.  
\par
Define the \emph{minimum clockwise distance} between $p_i$ 
and a token to be the smallest clockwise distance from $p_i$ to $p_j$ such that
$p_j$ has $t>0$ tokens.  Observe that if $p_i$ has a token, then the minimum 
clockwise distance to a token is zero.  Similarly, let the minimum 
counterclockwise distance from $p_i$ to a token be defined.  
Let $Rdist_i$ denote the minimum clockwise distance to a token for $p_i$ and
let $Ldist_i$ denote the minimum counterclockwise distance to a token for $p_i$.

\subsection{Programs}
\label{section:program}

The \textsf{delay} and \textsf{relay} programs are shown in Figure 
\ref{fig:bothlay}.  Both programs begin with steps to move any queued tokens
from the predecessor's queue to rest at $p_i$.   The 
\textsf{relay} program enqueues one token, if there are any resting tokens,
in line 3 of the program.  The \textsf{delay} program may or may not enqueue
a token, depending on values of the counter $c_i$ and the number of resting
tokens $r_i$.  In terms of a Petri net, the \textsf{relay} program corresponds to simple,
deterministic, unit delay with at most one token firing in 
any step on the output transition.   The \textsf{delay} program expresses
a joint transition, with two inputs and two outputs:  the variable $c_i$ becomes
a ring of $C+1$ places and line 4 of \textsf{delay} represents the joint 
transition. 
    
\begin{figure}[ht]
\framebox[\columnwidth][l]{
\begin{minipage}{0.5\columnwidth}
\begin{tabbing}
x \= xxxx \= xx \= xx \= xx \= xx \= xx \= \kill
\> \textsf{delay} :: \\
\>\> \texttt{do forever} \\
{\scriptsize 1} \>\>\> $r_i \;\leftarrow~ r_i + q_{i-1}$ ~\texttt{;} \\
{\scriptsize 2} \>\>\> $q_{i-1} \;\leftarrow~ 0$ ~\texttt{;} \\ \\
{\scriptsize 3} \>\>\> \texttt{if} \quad $c_i>0$ \quad \texttt{then} \\
\>\>\>\>\> $c_i\;\leftarrow~ c_i-1$ \\
{\scriptsize 4} \>\>\> \texttt{else if}  
$c_i=0 \;\wedge\; r_i>0$ \texttt{then}\hspace*{2em}\\
\>\>\>\>\> $c_i \;\leftarrow\; C$ ~\texttt{;} \\
\>\>\>\>\> $r_i \;\leftarrow\; r_i - 1$ ~\texttt{;} \\
\>\>\>\>\> $q_i \;\leftarrow\; q_i + 1$ \\ 
\end{tabbing} 
\end{minipage}
\vrule
\begin{minipage}{0.5\columnwidth}
\begin{tabbing}
xx \= xx \= xxxx \= xx \= xx \= xx \= xx \= xx \= \kill
\>\> \textsf{relay} :: \\
\>\>\> \texttt{do forever} \\
\> {\scriptsize 1} \>\>\> $r_i \;\leftarrow~ r_i + q_{i-1}$ ~\texttt{;} \\
\> {\scriptsize 2} \>\>\> $q_{i-1} \;\leftarrow~ 0$ ~\texttt{;} \\ \\
\> {\scriptsize 3} \>\>\> \texttt{if} \quad $r_i>0$ \quad \texttt{then} \\
\>\>\>\>\> $r_i \;\leftarrow~ r_i - 1$ ~\texttt{;} \\ 
\>\>\>\>\> $q_i \;\leftarrow~ q_i + 1$ \\ 
\end{tabbing}
\end{minipage}
}
\caption{\textsf{delay} and \textsf{relay} programs}
\label{fig:bothlay}
\end{figure}

In application, it is possible that all $n$ processes run the 
\textsf{delay} program, and no process runs \textsf{relay}.  This 
would be a \emph{uniform} protocol to achieve \desd{1}--\desd{5}.  
An advantage of including \textsf{relay} processes can be to limit
the cost of construction for physical embodiments of the logic.  
Using multiple \textsf{relay} processes can model 
more general cases of token delay:  a consecutive sequence of $k$ 
\textsf{relay} processes is equivalent to a process that always
delays an arriving token by $k$ rounds.

\subsection{Application to Models}
\label{section:appmodel}

\paragraph{Petri Net.}
It is usual for self-stabilization that transient faults, 
which inject variable corruption, are responsible for creating 
new initial states, and the event of a transient fault is not 
explicitly modeled.  However for an application where tokens 
represent physical objects, which is plausible for Petri nets, 
a transient fault neither destroys nor creates objects.  Thus 
we think it reasonable to suppose 
that $m>1$ tokens satisfying (\ref{eqn:separation}) are present 
in any initial state.  

Observe that line 2 of either \textsf{delay} or \textsf{relay} 
has $p_i$ assign to $q_{i-1}$ (whereas the usual convention in the
literature of self-stabilization is that a process may only assign
to its own variables).  The assignment $q_{i-1}\;\leftarrow\;0$ 
models the transfer of a token from a transition to its target 
place in a Petri net.  For the firing of a Petri net transition, 
$p_i$ increments $q_i$ in line 4 of \textsf{delay} or line 3 
or \textsf{relay}.  Figure \ref{figure:petri} illustrates both 
\textsf{relay} and \textsf{delay} programs.  The portions of the 
two rings on the left side of the figure are modeled by the $c$ 
variables in \textsf{relay} nodes;  these are ``minor'' rings with
$C+1$ nodes, whereas the ``major'' ring has $n$ nodes.  The situation 
of a token on a minor ring being ready for a transition shared by 
the major ring is modeled by $c_i=0$.  Observe that when a token  
on the major ring is present at the same transition where a minor ring
token exists, then transition firing is enabled at line 4, because 
$r_i>0$ and $c_i=0$.  We assume that tokens of major and minor rings are of 
different nature; a transient fault cannot move a token from minor to major
or from major to minor ring.  A transient fault can move tokens arbitrarily
on their respective rings.  

\paragraph{Standard Models.} We first briefly review some 
conventions from the literature of self-stabilization. 
A typical model for self-stabilizing protocols
is the shared variable model in which each process has access to
some variables written by neighbor processes. 
In one atomic step, a process reads neighbor variables, 
performs some local computation, and writes to its variables.  
A system execution is a sequence of configurations, each configuration
denoting the state of every process;  between each consecutive pair of
configurations in the execution there is a transition consisting of a
some set of process steps (at most one step for each process).  
Standard models specify a scheduler, which selects, at each state in an execution, the process or 
processes that may take a step.  The specification of the scheduler and the
set of possible initial states is enough to generate all possible executions.  
Schedulers may be synchronous (all processes
take a step in unison) or asynchronous;  an extreme case of asynchrony is
the central daemon scheduler, which selects just one process to take a 
step.  The usual notation for programs is the guarded statement notation, 
wherein the program for each process is a set of guarded assignment statements.
A guarded assignment is \emph{enabled} at a particular state if the guard
evaluates to true at that state.  To avoid stuttering (repeated consecutive configurations
in an execution), schedulers only select processes with enabled statements; 
also, programs are written so that any enabled statement should falsify its guard 
when executed  (this is typically easy to verify for the central daemon scheduler).  
An execution is finite if its terminal state has no enabled statement, otherwise executions are
infinite.  Schedulers may have a number of choices of processes to select
for the next step at a particular state.  A fairness property of a 
scheduler is some policy to limits choices it makes over the course of an execution.
An unfair scheduler has maximum freedom in the number and guard    
selection choice at any state.  Experience has shown that programs are 
simplified when more assumptions can be made about the scheduler;  for some
problems, self-stabilization is not possible without the central daemon 
hypothesis of one process stepping in any state transition.  
\par
Considering how \textsf{delay/relay} may be fit to standard models, 
we see several obstacles: \itp{i} $p_i$ 
assigns to $q_{i-1}$, which violates the rule of a process assigning to 
its own variables only; \itp{ii} we have assumed that all process start
their cycles together, which may not hold for an initial configuration; 
\itp{iii} the number of tokens $m$ is supposed 
positive and constrained by (\ref{eqn:separation}) in the initial state; and
\itp{iv} execution is synchronous.  Point \itp{iv} is within the bounds
of self-stabilization models, though one might hope for a realization of 
the same result for asynchronous models.  Point \itp{ii} will not be a 
concern if the programs \textsf{delay} and \textsf{relay} can each be 
reduced to a single guarded assignment;  this is not a significant challenge,
and we leave this as an exercise to the reader.  We continue examining
the other points in following paragraphs.
\par
Regarding \itp{i}, there are two cases to consider, a synchronous or 
asynchronous execution model.  In the case of synchronous execution, rewriting
the program as a single guarded assignment can eliminate the assignment 
to $q_{i-1}$ in favor of having $p_i$ rewrite $q_i$, either to zero or to 
some new value if a token is queued.  Because each process reads $q_{i-1}$ 
in each synchronous step, the logic of \textsf{delay} and \textsf{relay} is
preserved by this rewriting.  However, for an asynchronous model, some 
transformation is needed.  For instance, a self-stabilizing protocol with 
acknowledgment \cite{AB93} might be used to convey a token from $p_{i-1}$ to $p_i$ 
(note that this approach would entail bidirectional communication between
$p_{i-1}$ and $p_i$).  Alternatively, the task of passing a token from 
$p_{i-1}$ to $p_i$ could be handled by using a conventional self-stabilizing
protocol, which we explain next under the discussion of point \itp{iii}. 
\par
Point \itp{iii} raises the possibility that the initial state may not 
have $m>0$ tokens satisfying (\ref{eqn:separation}).  Two ways to deal with
this possibility are \emph{active monitoring} and \emph{definitional} 
approaches.  
\begin{itemize}
\item The idea of active monitoring is to periodically sample the 
number of tokens and take appropriate measures for an incorrect value.  
Note that taking a sample is neither instantaneous nor reliable.  A sample
would need count the number of tokens that arrive to $p_i$, which should
be $m$ tokens over $n$ time units when behavior is legitimate.  This type
of sampling is unreliable from an arbitrary initial state, because whatever
variables are used for counting and measuring time are subject to transient
fault corruption, hence false detection of an illegitimate state is 
possible.  Moreover, if more than one process engages in sampling and 
correction, it could be that correction by one interferes with correction
by another.  The mechanism of distributed reset \cite{AG94} might need to 
be employed for active monitoring and correction.  In addition to the 
complexity of active monitoring, the potential for inserting and deleting 
tokens (perhaps unnecessarily) during convergence could be undesirable for
the application using the tokens.  
\item The alternative to active monitoring is
the definitional approach.  Here, the system is built either from $m$ 
independent self-stabilizing token rings or from a self-stabilizing multitoken
ring protocol that has $m$ tokens in a legitimate state.  We sketch the case
of $m$ independent token rings.  For each token ring, there can be more than
one token in an initial state.  Provided that each process fairly includes
steps from each of the $m$ token rings, each of these eventually converges
to having a single token.  A benefit of the definitional approach is that 
the token-passing mechanism is unidirectional:  by writing to some variable
designated for the token, the token is automatically available to the 
successor process.  
\end{itemize}
\par  
Can the \textsf{delay/relay} protocol be extended to asynchronous scheduling?
To reckon with \itp{iv}, some relaxation of \desd{2}-\desd{3} is needed, because
no mechanism in an asynchronous model can assure that two processes at distance
$d$ release tokens simultaneously.  A natural adaptation of the synchronous
protocol is to leverage a self-stabilizing synchronizer, or a self-stabilizing
phase clock. 
\begin{quote}
A phase clock protocol equips each process $p_i$ with a \emph{clock} variable  
$clock_i$.  A clock has domain $[0,M]$, where $M$ has a lower bound related to the diameter 
of the network, but can otherwise be freely chosen;  we suppose $M\bmod d = 0$ for 
our design.  The two crucial properties of
a clock are \itp{a} it increments modulo $M$ infinitely often in any execution, 
and \itp{b} $|clock_i-clock_j|\leq 1$ for neighbors $i$ and $j$.  
\end{quote}
An adaptation based on a phase clock consists of allowing a cycle (translated
into a guarded assignment) of \textsf{delay/relay} at $p_i$ to execute only when 
$clock_i\bmod d = 0$ and the phase clock at $p_i$ is enabled to increment.  
The phases where $clock_i\bmod d = \ell$, for $\ell\in[1,d-1]$ are ``idle''
with respect to progress of \textsf{delay/relay}.  
Thus, $p_i$ can only release a token when $clock_i\bmod d = 0$.   Property 
\itp{b} implies that a process $p_k$ at distance $d$ from $p_i$ satisfies
$|clock_i-clock_j|\leq d$.  Therefore, $clock_k$ could be ``behind'' 
$clock_i$ by $d$ phases when $p_i$ releases a token.  Increments to the phase
clock at $p_i$ continue during idle phases.  The modulus $d$ provides sufficiently 
many idle phases to ensure that $p_k$ would release a token (provided it has one 
ready to release), before $p_i$ again encounters $clock_i\bmod d=0$.  
\par
Two difficulties need to be addressed in implementing this combining of 
\textsf{delay/relay} with a phase clock.  First, we note that  
\desd{2} can be violated in the combination, because $p_i$ may release a token
before $p_k$ does, resulting in two tokens at distance $d-1$; some revision to 
$d$ or \desd{2} is needed to handle this case.  Another detail 
to cover in the adaptation is to prevent $p_{i+1}$ from immediately passing the
token it receives from $p_i$, which would occur if $clock_{i+1}\bmod d=0$ upon 
reception.  We omit further detail in this outline of adapting \textsf{delay/relay} 
to an asynchronous scheduler.

\paragraph{Wireless Sensor Network.}  
Wireless networks use messages rather than shared variables for communication.  
Several papers have proposed some implementation patterns for a shared variable
model built upon sensor network abstractions.  However, the definitional approach
outlined above for tokens is well-suited to a message-passing architecture because
releasing a token consists of a single write to a shared variable that is not 
again written until the next time a token is released.  This allows the write to 
the shared variable to be replaced by a message transmission, from $p_i$ to 
$p_{i+1}$ (such a message would contain values for all $m$ independent 
token rings being emulated in the definitional construction).  
\par
Wireless sensors do have real-time clocks, though these may not be synchronized
in an initial state.  Therefore, a self-stabilizing clock synchronization protocol
is warranted, so that the synchronous execution of \textsf{delay/relay} can be 
realized.  With synchronized clocks, it is possible to define a recurring time
interval so that all nodes start and end the execution of a \textsf{delay/relay}
cycle together.  
\par
The preceding discussion assumes that messages are reliable.  In practice, 
messages may be lost, which necessitates retransmission.  The number of 
retransmits could be variable, which is problematic for defining intervals
supporting synchronous execution of cycles.  Of course, any protocol for wireless
sensor networks faced with message loss is, at best, probabilistically valid.    

\subsection{Verification}
\label{section:verification}

A legitimate state for the protocol is a global state predicate, 
defining constraints on values for variables. 
To define this predicate, let $tokdist$ denote the minimum, 
taken over all $i$ such that $r_i+q_i>0$, of $Rdist_i$.  
The predicate $\textsf{delay}_i$ is \emph{true} for process 
$p_i$ running \textsf{delay} and \emph{false} for the 
\textsf{relay} processes.
\begin{definition}
A global state $\sigma$ is legitimate iff 
\begin{eqnarray}
 \sum_i q_i = m & \quad\wedge\quad & \sum_i r_i = 0 \quad\wedge\quad (\forall i:: \; q_i\leq 1) \label{leg:0} \\
 &\;\wedge\;& tokdist > C \label{leg:1} \\
 &\;\wedge\;& \left(\forall i: \;\textsf{delay}_i\;\wedge\; 
   c_i>0 \;\wedge\; q_i=0: \; 
   Rdist_i=C-c_i\right) \label{leg:2} \\ 
 &\;\wedge\;& \left(\forall i: \;\textsf{delay}_i \;\wedge\;q_i=0: 
   \; Ldist_i > c_i\right) \label{leg:3} \\
 &\;\wedge\;& \left(\forall i: \;\textsf{delay}_i \;\wedge\;q_i=1: 
   \; c_i=C\right) \label{leg:4} 
\end{eqnarray}
\end{definition}   
In an initial state, variables may have arbitrary values in their domains, 
subject to constraint (\ref{eqn:separation}).  

\begin{lemma}[Closure] \label{lemma:closure0}
Starting from a legitimate state $\sigma$, 
the execution of a round results in a legitimate state $\sigma'$. 
\end{lemma}
\begin{proof}  
The conservation of tokens expressed by (\ref{leg:0}) is simple to 
verify from the statements of \textsf{delay} and \textsf{relay} programs, 
so we concentrate on showing that (\ref{leg:1})--(\ref{leg:4}) are 
invariant properties.  Assume that $\sigma$ is a legitimate state. 
We consider two cases for a process $p_i$ running \textsf{delay}, 
either there is no token at $p_i$ and $q_i=0$, or $q_i=1$ at $\sigma$.  

$\blacktriangleright q_i=1:$ observe that $c_i=C$ by (\ref{leg:4}).  
For $\sigma'$ we have $r_i=0$ because from (\ref{leg:1}) there is no 
token at $p_{i-1}$ and we have $q_i=0 \;\wedge\; c_i=C-1$ by lines
1-2 of either \textsf{delay} or \textsf{relay} at $p_{i+1}$, and 
line 3 of \textsf{delay} at $p_i$.  This validates (\ref{leg:0}) with
respect to the token passed, and (\ref{leg:1}) holds because every 
token moves to the successor starting from a legitimate state.  
Property (\ref{leg:2}) holds at $\sigma'$ with respect to $p_i$ because
$Rdist_i=1=C-c_i$.  Finally, (\ref{leg:3}) is validated for $p_i$ 
because, if (\ref{leg:3}) holds for $\sigma$ when $c_i=C$, then a 
token moving one process closer to $p_i$ validates (\ref{leg:3}) 
by $c_i=C-1$ at $\sigma'$.   

$\blacktriangleright q_i=0:$ there are two subcases, either 
$c_i=0$ or $c_i>0$.  In the former case, if no token arrives to 
$p_i$ in the transition from $\sigma$ to $\sigma'$, 
properties (\ref{leg:0})-(\ref{leg:4}) directly hold with respect to $p_i$ 
in $\sigma'$.  If a token arrives to $p_i$, then $q_i=1\;\wedge\;c_i=C$ 
result by line 4 of \textsf{delay}, and we use properties 
(\ref{leg:1})-(\ref{leg:4}) of $p_{i-1}$ at $\sigma$ to infer that
the same properties hold of $p_i$ at $\sigma'$.  If $c_i>0$ at 
$\sigma$, then by (\ref{leg:1})-(\ref{leg:3}) and legitimacy of 
all processes within distance $C+1$ in either direction from $p_i$, 
tokens move to the successor process while $c_i$ decrements, 
which establishes (\ref{leg:1})-(\ref{leg:3}) for $p_i$ at 
$\sigma'$. 
\end{proof}

To prove convergence, we start with some elementary claims and 
define some useful terms. 
Suppose rounds are numbered in an execution, round $t$ starts
from state $\sigma$, and that $q_{i-1}^{\sigma}=v$.  For such a situation,
we say that $v$ tokens \emph{arrive} at $p_i$ in round $t$.  

\begin{claim} \label{claim:convX0} 
In any execution, $\sum_i r_i+q_i \;=\;m$ holds invariantly.
\end{claim}
\begin{proof}  As explained in Section \ref{section:model},
$m>1$ tokens are present in the initial state, represented by 
$r_i$ and $q_i$ variables.  Statements of \textsf{delay} or 
of \textsf{relay} conserve the number of tokens in the system, 
because we assume that all processes execute lines 1 and 2 synchronously 
in any round.  Line 4 of \textsf{delay} or line 3 of \textsf{relay}
similarly conserve the number of tokens in a process.
\end{proof}
\begin{claim} \label{claim:conv0}
Within one round of any execution, 
\begin{eqnarray}
&& ( \; \forall \, i:: \quad q_i\leq 1 \;)  \label{cla:0}
\end{eqnarray}
holds and continues to hold invariantly for all subsequent rounds.
\end{claim}
\begin{proof}
In every round, line 2 of \textsf{delay} or \textsf{relay} assigns 
$q_{i-1}\leftarrow 0$, and may assign $q_i\leftarrow 1$.
\end{proof}

For the remainder of this section, we consider only executions 
that start with a state satisfying (\ref{cla:0}).  For such 
executions, a corollary of Claim \ref{claim:conv0} is:  at most one token 
arrives to any process in any round.  

\begin{claim} \label{claim:conv1} 
If $m>0$, then for every execution of the protocol 
and $0\leq k<C \;\wedge\; 0\leq i<n$, the variable $c_i=k$ 
at infinitely many states.
\end{claim}
\begin{proof}
The proof is by contradiction.  First, we show that at least some token moves
infinitely often.  Since $m>0$, there is a token at some process $p_i$ because
$r_i>0$ or $q_i>0$;  the case $q_i>0$ implies immediate token movement in
the next round, so we look at the other case, $r_i>0\;\wedge\;q_i=0$.    
In one round, $p_i$ either assigns $q_i\leftarrow 1$ program, 
and thus a token moves in the next round, or $p_i$ assigns $c_i\leftarrow 
c_i-1$ because $c_i>0$.  Therefore, after at most $C$ rounds, a state 
where $c_i=0\;\wedge\;r_i>0$ is reached, and the next round enqueues a 
token for movement.  The preceding argument shows that \emph{some} token
movement occurs infinitely often in any execution.  There are $m$ tokens
throughout the execution by (\ref{leg:0}), hence at least one token can be
considered to move infinitely many times.  Note that the token abstraction
is represented by $r$ and $q$ variables only:  we cannot be sure that one
token does not overtake another token.  However, it will be our convention
to model token queuing as first-in, first-out order.  Thereby we find that
if one token moves infinitely many times clockwise around the ring, all 
tokens do so as well.  Returning to the claim, suppose that some $c_i$ 
variable eventually never has some value $k\in[0,C-1]$.   But $p_i$ experiences
infinitely many tokens arriving and assigns $q_i\leftarrow 1$ infinitely 
often, so lines 3 and 4 of \textsf{delay} execute infinitely often at $p_i$, 
which is a contradiction. 
\end{proof}
\begin{claim} \label{claim:conv2} 
For any process $p_i$ running the \textsf{delay} program, eventually 
$p_i$ assigns $q_i\leftarrow 1$ at most once in any $C+1$ consecutive rounds. 
\end{claim}
\begin{proof}
Claim \ref{claim:conv1} shows that $p_i$ eventually assigns $c_i\leftarrow C$.  
In \textsf{delay}, only line 4 assigns $q_i\leftarrow 1$, and the same
step assigns $c_i\leftarrow C$.  Thus, if we number rounds $t$, $t+1$, and 
so on, the values of $c_i$ and $q_i$ variables \emph{at the end} of each 
round is shown by:
\par\begin{center}
\begin{tabular}{|r|c|c|c|c|c|c|c|c|c|}
\hline 
\textit{round} & $t$ & $t+1$ & $t+2$ & $\cdots$ & $t+(C-1)$ & $t+C$ & $t+(C+1)$ \\ \hline
$c_i$ & $C$ & $C-1$ & $C-2$ & $\cdots$ & 1 & 0 & $C$ \\ \hline
$q_i$ & 1 & 0 & 0 & $\cdots$ & 0 & 0 & 1 \\ \hline 
\end{tabular}
\end{center}\par 
The table makes the worst-case assumption that $r_i>0$ at round
$t+(C+1)$, to illustrate that through at least $C$ consecutive rounds,
$q_i$ remains zero.
\end{proof}
\begin{claim} \label{claim:conv3} 
Let $\sigma$ be a state that occurs after sufficiently many rounds so 
that every token has arrived at least once at some process running
the \textsf{delay} program.  In the (suffix) execution following $\sigma$, 
if a token arrives at $p_i$ in round $t$, then no token arrives at 
$p_i$ during rounds $t+1$ through $t+C$.   
\end{claim}
\begin{proof}
After $\sigma$, a token departs from any given \textsf{delay} process
$p_i$ only at line 4 of \textsf{delay}, which has precondition $c_i=0$ 
and postcondition $c_i=C$.  Thus, each \textsf{delay} process releases
a token once every $C+1$ rounds (or less often, if $r_i=0$ holds).   
A \textsf{relay} process $p_j$ could potentially release tokens once 
per round, if $r_j$ remains positive, however we have supposed that 
each token has entered some \textsf{delay} process before $\sigma$.  
A simple inductive argument shows that $r_j=0$ holds throughout the 
execution following $\sigma$.  Therefore, each process experiences 
token arrival at most once every $C+1$ rounds.  
\end{proof}

Below we consider only executions 
that start with a state $\sigma$ satisfying (\ref{claim:conv3}).  For such 
executions, a corollary of Claim \ref{claim:conv3} is:  the token 
arrival rate to any process is at most $1/(C+1)$.    

\begin{claim} \label{claim:conv4}
Let $\sigma$ be a state in any execution identified by Claim \ref{claim:conv3}.
Then, throughout the remainder of the execution, 
\begin{eqnarray}
&& ( \; \forall \, i:: \quad r_i\leq 1+r_i^{\sigma} \;)  \label{cla:1}
\end{eqnarray}
\end{claim}
\begin{proof} If some $r_i$ increases in a round, 
then because of Claim \ref{claim:conv3}, no additional token arrives to $p_i$
for $C+1$ rounds; this is an adequate number of rounds to ensure 
that $p_i$ will release a token, thus putting $r_i$ back to its original value.
\end{proof}
Claim \ref{claim:conv4} allows us to introduce the notion of a 
\emph{resting bound} for any process $p_i$.  With respect an execution
$E$ with an initial state $\sigma$ as defined in Claim \ref{claim:conv3}, 
for each $p_i$ executing the \textsf{delay} program there is 
a bound $1+r_i^{\sigma}$ on the number
of tokens that may rest together at $p_i$ during the remainder of $E$.
After any round in $E$ producing a state $\beta$, 
let us consider three possibilities for any particular 
\textsf{delay} process $p_i$'s number of resting tokens:    
\[ r_i^{\beta}=(1+r_i^{\sigma})\quad\vee\quad 
   r_i^{\beta}=r_i^{\sigma}\quad\vee\quad
   r_i^{\beta}<r_i^{\sigma} \]
Notice that for the first two disjuncts, the resting bound of $r_i$ is 
unchanged.  However, for the third disjunct, where $p_i$ released a token
in the round and has fewer than $r_i^{\sigma}$ resting tokens, the argument
of Claim \ref{claim:conv4} can be applied to state $\beta$, lowering the
resting bound for $p_i$ to $1+r_i^{\beta}$.  More generally, 
there might be other processes that lower their resting
bounds during the round obtaining $\beta$.  Thus, the resting
bound developed by Claim \ref{claim:conv4}  
for each process may improve during an execution.  At any particular
point in $E$, the best bound for $r_i$ is $1+r_i^{\delta}$, 
where $\delta$ is determined by the most recent round that lowered
$r_i$'s resting bound; if there is no such preceding round, then 
let $\delta=\sigma$.  Below, we find special cases with more accurate bounds.
\par
Resting bounds are the basis for a variant function on 
protocol execution.  Let $F$ be a tuple formed by listing the 
resting bounds for all \textsf{delay} processes.  
Since no process increases its resting bound in any round, it follows that
valuations of $F$ can only decrease during execution.  If all components 
of tuple $F$ are zero, then it is possible to show that the protocol has reached
a legitimate state.  We say that $F$ is positive if any of its components is 
nonzero.  In order to prove convergence, two more claims are
needed.  First, a special case is required for a resting bound of zero 
(since Claim \ref{claim:conv4}'s form is inappropriate);  second, it 
must be shown that $F$ eventually does decrease if it is positive.
\begin{claim} \label{claim:conv5}
Let $E$ be an execution originating from a state
$\sigma$ identified by Claim \ref{claim:conv3},
and suppose $\alpha$ is a state in $E$ where 
$p_i$ running \textsf{delay} satisfies 
$r_i^{\alpha}=0\;\wedge\;c_i^{\alpha}=0$.  
Then, for the execution following $\alpha$, the
resting bound of $r_i$ is zero.
\end{claim}
\begin{proof} 
The proof is by induction over the execution following
$\alpha$, based on the sequence of rounds associated with
token arrival at $p_i$.
When a token arrives at $p_i$ after state $\alpha$, the 
predicate $r_i=0\;\wedge\;c_i=0$ holds.  The \textsf{delay}
program establishes $q_i=1\;\wedge\;r_i=0\;\wedge\;c_i=C$ when
processing the arriving token.  Claim \ref{claim:conv3} ensures
that no additional token will arrive to $p_i$ during the following
$C$ rounds, so that $c_i=0$ holds when the next token arrival 
occurs for $p_i$.  
\end{proof}
\begin{claim} \label{claim:conv6}
Let $E$ be an execution originating from a state
$\sigma$ identified by Claim \ref{claim:conv3}.
$F$ cannot be positive and constant throughout $E$.
\end{claim}
\begin{proof}
Proof by contradiction.  Suppose, for each 
\textsf{delay} $p_i$, that the resting bound never
decreases.  We analyze scenarios for this supposition
and derive necessary conditions, which are used to 
show a contradiction.  Claim \ref{claim:conv1} implies that
$p_i$ receives a token infinitely often in $E$ and
releases a token infinitely many times.  If 
each token reception coincides with releasing a 
token, which entails $c_i=0$, then $r_i>0$ would 
remain constant throughout $E$.  For such a 
continuing scenario, $p_i$ must receive at token 
once every $C+1$ rounds.  The other possible 
scenario is that of $r_i$ incrementing to the 
resting bound, then decrementing, and repeating this
pattern.  For this scenario, $c_i$ decrements to 
zero, then resets to $C$, continuously during $E$.  
Because we have supposed that $r_i$ never decrements
twice before incrementing again (otherwise $F$ would
decrease), tokens need to arrive sufficiently often
to $p_i$.  Suppose $c_i=k>0$ when a token arrives.  
Then another token must arrive after exactly $C+1$
rounds so that $c_i=k$ upon token arrival.  If instead,
a token arrives after $C+1+d$ rounds, then $c_i=k-d$ 
would hold upon token arrival, $d\leq k$.  More generally,
the spacing in token arrivals over $E$ could be $C+1+t_1$, then 
$C+1+t_2$, and so on up to a delay gap of $C+1+t_{\ell}$ rounds, 
so long as $\sum_{j=1}^{\ell} t_j \leq k$.  Each time there
is a delay gap of $C+1+t_j$ rounds with $t_j>0$, the 
counter $c_i$ decreases in this scenario.  A decrease resulting
in $c_i=0$ would then force all future delays to be $C+1$, so
that tokens arrive exactly as they are released, preventing a
reduction of $r_i$.     
\par
Having exposed the scenarios for $F$ remaining constant
over execution $E$, thus at least one $r_i>0$ throughout 
$E$, we observe that from $m(C+1)\leq n$ there exists 
a segment of the ring, with more than $C+1$ processes, 
containing no token, at every state in $E$.  The existence
of such a segment implies that some \textsf{delay} $p_i$ 
will receive a token after a delay of more than 
$C+1$ rounds (a more detailed argument could take into
account processes identified by Claim \ref{claim:conv5}, 
which merely pass along tokens when they arrive).  
Thus, infinitely often, a delay between token arrivals 
is at least $C+2$ rounds.  It follows that eventually,
$c_i$ decreases to zero for some $p_i$ before a token
arrives, at which point it will decrease $r_i$, 
contradicting the assumption that $F$ remains constant.
\end{proof}
\begin{lemma}[Convergence] \label{lemma:converge0}
Every execution of the protocol eventually contains a legitimate state.
\end{lemma}
\begin{proof}
The proof is by induction on Claim \ref{claim:conv6}, so long as $F$ is 
positive.   Therefore, the resting bound for every 
\textsf{delay} process is zero eventually.  
To establish that the resulting state
is legitimate, it is enough to verify the behavior of a \textsf{delay}
process $p_i$ during the $C+1$ rounds preceding token arrival to 
see that (\ref{leg:1})-(\ref{leg:3}) hold with respect to $p_i$. 
\end{proof}

We sketch an argument bounding 
the worst-case convergence time using elements from the 
proof of convergence.  The variant function $F$ is applied to an 
execution suffix satisfying (\ref{cla:0}), and Claim \ref{claim:conv3};
such a suffix occurs within $O(n)$ rounds of any execution: 
the worst case occurs when one \textsf{delay} process $p_i$ holds all $m$ 
tokens, which it releases after at most $m\cdot C$ rounds, and the last
of these tokens takes $O(n)$ rounds to again arrive at a \textsf{delay}
process;  since $m\cdot C\leq n$ by (\ref{eqn:separation}), 
we have $O(n)$ rounds overall to obtain the suffix for Claim \ref{claim:conv3}.  
To bound the worst case for $F$ reducing in an execution, we observe
that there are at most $n$ components to $F$, each with an initial maximum
value of $m$.  Suppose each component decreases sequentially, therefore 
requiring $n\cdot m\cdot f$ time, where $f$ is the worst-case number of rounds
to reduce one component of $F$.  We bound $f$ by the proof argument of 
Claim \ref{claim:conv6}.  A ring segment of length at least $C+2$ and 
devoid of tokens implies some decrease of a resting bound in the proof 
argument.  This decrease may take $O(C)$ time to occur, as a $c$ variable
reduces while a process awaits a token.  If $f\in O(C)$, an 
overall bound on convergence time is $O(n)+O(n\cdot m \cdot C) = O(n^2)$;  
however, if the ring segment without tokens is longer, then a \textsf{delay}
process may spend more time awaiting a token.  A conservative bound is 
therefore $O(m\cdot n^2)=O(n^3)$ rounds.  

As an aside, we note that the algorithms are deterministic, execution is 
fully synchronous (there is no nondeterministic adversary), and the 
program model fits the Petri net formalism; therefore a formulation
using the max-plus algebra \cite{FCOQ92} can express system 
execution, and there exist tools to compute eigenvalues for a matrix
representing the system.  We did not investigate such an approach, since
the choice of which processes run \textsf{delay} would be an extra 
complication.  

\begin{figure}[htb]
\quad
\begin{minipage}{0.45\columnwidth}
\input{graph1.tex}
\end{minipage}
\qquad
\begin{minipage}{0.45\columnwidth}
\input{graph2.tex}
\end{minipage}
\begin{quote}
{\footnotesize
Each point is derived from 300 simulations with a random initial
state and random selection of which processes run \textsf{delay}, ranging
from one to 50 on the $x$-axis;  the $y$-axis is the average number of rounds 
for convergence. 
}
\end{quote}
\caption{Simulations with $n=50$, $m=2$ and $m=5$}
\label{figure:simulation1}
\end{figure}
\begin{figure}[htb]
\quad
\begin{minipage}{0.45\columnwidth}
\input{graph3.tex}
\end{minipage}
\qquad
\begin{minipage}{0.45\columnwidth}
\input{graph4.tex}
\end{minipage}
\begin{quote}
{\footnotesize
The graph on the left continues the type of simulation shown in 
Figure \ref{figure:simulation1}, however with 10 tokens.  
The graph on the right varies the ring size, $n=10i$ for 
$0<i<20$, in each case letting $m=2$, $d=\lfloor n/3\rfloor$,
and the number of \textsf{delay} processes be $n/2$.
}
\end{quote}
\caption{Other simulations}
\label{figure:simulation2}
\end{figure}
We have simulated the protocol for various cases of $n$, $m$, 
$C$, and choices for the number of \textsf{delay} processes.  These 
simulations suggest that a bound $O(n^3)$ may be loose for the average 
case.  Another point of the simulation is to investigate the influence 
of having multiple \textsf{delay} processes on convergence time.  
Our simulations explored random initial values for variables. Three graphs,
spread over figures \ref{figure:simulation1} and \ref{figure:simulation2}, 
experiment with differing values of $m$ and $d$, all for a 50-node ring.
In each graph, an experiment is repeated for different numbers of 
\textsf{delay} processes, from 1 to 50.  The results suggest that 
having at least a few \textsf{delay} processes is beneficial.  To 
explore another dimension, the ring size $n$, a fourth graph presented
in Figure \ref{figure:simulation2} varies $n$:  the results suggest
that expected convergence time is linear in $n$. 
\begin{theorem} \label{theorem:selfstab}
The \textsf{delay/relay} protocol, with at least one \textsf{delay}
process, and with $n>1$, $m>1$, $d=C+1$, and $m\cdot d\leq n$, self-stabilizes
to desiderata \desd{1}-\desd{5}.
\end{theorem}
\begin{proof}
Claim \ref{claim:convX0} attends to \desd{1}, 
The definition of legitimate state validates \desd{2}.  A property of 
a legitimate state is that $c_i=0$ whenever a token arrives to $p_i$, hence
the behavior of \textsf{delay} is like \textsf{relay}:  a token moves
in each round, as required by \desd{3}.  The ring topology and the unidirectional
movement of tokens satisfies \desd{4}.  Finally,   
lemmas \ref{lemma:converge0} and \ref{lemma:closure0} provide the 
technical basis for the theorem, showing \desd{5}.
\end{proof}

\section{Protocol with Unknown Ring Size}
\label{section:unknown-sep}

A parameter $C$, upon which the target separation between
tokens is based, is given to the protocol of Section
\ref{section:known-sep}.  Here we consider another design
alternative, where the separation between tokens should be
maximized, but the ring size is unknown.  The technique is
straightforward:  building upon the \textsf{delay} program,
additional variables are added to count the number of rounds
needed to circulate a token, that is, the new program
calculates $n$.  Two extra assumptions are used for the
new protocol: the value of $m$ is known and the number of
processes running the \textsf{delay} program is exactly one.
We discuss this limitation in Section \ref{section:conclusion}.

\begin{figure}[ht] \hspace*{0.15\columnwidth}
\framebox[0.7\columnwidth][c]{
\begin{minipage}{0.7\columnwidth} 
\begin{tabbing} x \= xxxx \= xx \= xx \= xx \= xx \= xx \= \kill 
\> \textsf{delay} ::
\\ \>\> \texttt{do forever} \\ 
{\scriptsize 1} \>\>\> $t_i \;\leftarrow\; t_i + 1$ ~\texttt{;} \\ 
{\scriptsize 2} \>\>\> \texttt{if} $q_{i-1}>0\;\wedge\;timing_i$ \texttt{then} \\
{\scriptsize 3} \>\>\>\> $ignore_i\;\leftarrow~ ignore_i-1$ ~\texttt{;} \\ 
{\scriptsize 4} \>\>\>\> \texttt{if} ~$ignore_i< 1$~ \texttt{then} \\ 
{\scriptsize 5} \>\>\>\>\> $timing_i \;\leftarrow~ \textit{false}$ \texttt{;} \\ 
{\scriptsize 6} \>\>\>\>\> $ClockBase_i \;\leftarrow~ \lfloor t_i/M\rfloor -1$ \\ \\ 
{\scriptsize 7} \>\>\> $r_i \;\leftarrow~ r_i + q_{i-1}$ ~\texttt{;} \\ 
{\scriptsize 8} \>\>\> $q_{i-1} \;\leftarrow~ 0$ ~\texttt{;} \\ \\ 
{\scriptsize 9} \>\>\> \texttt{if} \quad $c_i>0$ \quad \texttt{then} \\ 
{\scriptsize 10} \>\>\>\> $c_i\;\leftarrow~ c_i-1$ \\
{\scriptsize 11} \>\>\> \texttt{else if} $c_i=0 \;\wedge\; r_i>0$ \texttt{then}\hspace*{2em}\\ 
{\scriptsize 12} \>\>\>\> $c_i \;\leftarrow\; ClockBase_i$ ~\texttt{;} \\ 
{\scriptsize 13} \>\>\>\> $r_i \;\leftarrow\; r_i - 1$ ~\texttt{;} \\
{\scriptsize 14} \>\>\>\> $q_i \;\leftarrow\; q_i + 1$ ~\texttt{;} \\ 
{\scriptsize 15} \>\>\>\> \texttt{if} $\neg timing_i$ \texttt{then} \\ 
{\scriptsize 16} \>\>\>\>\> $timing_i \;\leftarrow\; \textit{true}$ ~\texttt{;} \\
{\scriptsize 17} \>\>\>\>\>  $t_i \;\leftarrow\; 0$ ~\texttt{;} \\ 
{\scriptsize 18} \>\>\>\>\>	$ignore_i \;\leftarrow\; M-r_i-1$ 
\end{tabbing} 
\end{minipage} } 
\caption{\textsf{delay}
program revised to calculate ring size} \label{fig:newdelay}
\end{figure}

Figure \ref{fig:newdelay} presents the revised \textsf{delay}
program, which introduces $timing_i$, $t_i$, $ignore_i$,
and $ClockBase_i$.  The program uses $ClockBase_i$ in
place of parameter $C$, which is periodically recalculated.
The method of calculation relies upon knowing $M$ and knowing
that all other processes run \textsf{relay}.  The program
begins a timing phase in lines 16-18, which starts a counter
$t_i$ at zero, and calculates the number of tokens that are
elsewhere in the ring, $ignore_i$.  Subsequently, lines 3-6
handle token arrival for purposes of calculating ring size;
after $ignore_i$ arriving tokens are ignored, the next token
is the one that was released when the timing phase began.
Of course, this calculation can be incorrect in early rounds
of an execution, but eventually each timing phase culminates
in $t_i$ having the ring size at line 6.

\begin{lemma} \label{lemma:selfstab} With the \textsf{delay}
program of Figure \ref{fig:newdelay} at one process
and \textsf{relay} at all other processes, the system is
self-stabilizing to $C=\lfloor n/m\rfloor - 1$.  \end{lemma}
\begin{proof} By arguments below, in any execution,
$ClockBase_i=\lfloor n/m\rfloor -1$ holds throughout
a suffix execution.  Lemmas \ref{lemma:closure0} and
\ref{lemma:converge0} then apply to verify self-stabilization.
\par Let $p_k$ be the sole \textsf{delay} process.
For convergence, it is enough to show that $ClockBase_k$
obtains the maximum feasible value for $m$ tokens, that is,
$ClockBase_k=\lfloor n/m\rfloor -1$ holds throughout a suffix
of any execution.  The proof hinges on two cases, either
\itp{1} we have $(\forall i:\; i\neq k:\; r_i=0)$, or \itp{2}
some tokens rest at \textsf{relay} processes.  In case \itp{1},
because $p_k$ releases at most one token per round, and because
all \textsf{relay} processes pass along the token it receives
in the next round, it follows that \itp{1} holds invariantly.
Provided $m>0$, process $p_k$ infinitely often receives and
releases a token in any execution, so lines 5-6 and lines
16-18 of Figure \ref{fig:newdelay} are executed repeatedly.
Line 18 calculates one fewer than the number of tokens that
do not rest at $p_k$ at the instant a token is released from
$p_k$ to $p_{k+1}$.  Provided \itp{1} holds, it will be $n$
subsequent rounds before this token circulates the ring
and returns to $p_k$.  Lines 1-6 compute the elapsed time
between the release of this token and its return, so that
$t_k=n$ when line 6 calculates the value of $ClockBase_k$,
and this drives convergence in the remainder of the execution.
Case \itp{2} eventually disappears, because $ClockBase_k>1$
is calculated in each execution of line 6.  Resting tokens for
\textsf{relay} processes therefore do not persist, assuming
$m<n$: no process receives a new token in every round, hence
any positive number of resting tokens reduces to zero, from
whence the count of resting tokens cannot rise.  \end{proof}

\section{Discussion}
\label{section:conclusion}

This paper provides fault tolerant constructions for a timing behavior
in which $m$ loci of control are separated.  The program mechanisms are 
simple:  tokens carry no data and processes use few variables.   The 
first construction can be uniform, distinguished (with one unique 
corrective process), or hybrid.  The second construction requires one
distinguished process.  

An interesting question is whether there can be a hybrid or 
uniform protocol when the ring size and separation constant are unknown.  
For the style of algorithm in Section \ref{section:unknown-sep}
we conjecture the answer is negative.  If one \textsf{delay} process
$p_i$ has an accurate estimate for maximum separation $d=c_i+1$ and 
does not delay any arriving token, another process $p_j$ may have 
either a larger, inaccurate estimate, or may perceive that tokens 
are unaligned with its counter and therefore delay some arriving 
tokens.  Such delay would lead to $p_i$ detecting an apparently larger
ring size, since the measured traversal time around the ring would include 
$p_j$'s delays.  Hence $p_i$ would raise its estimate for the 
separation value.  Note that the problem may admit other types of algorithms:
for example, if tokens are allowed to carry data, this would enable 
processes to communicate.  Whether such increased communication power 
is useful is an open question.  Another direction would be to use 
randomized timing, so that different \textsf{delay} processes do not 
interfere.

The program of Section \ref{section:known-sep} conforms to the standard
Petri net model of behavior control if we replace counters by auxiliary 
token rings, as shown in Figure \ref{figure:petri}.  This restriction 
enables tokens to model a physical system.  However, programs that use
tokens to carry data and thus communicate with explicit data rather than
mere timing of tokens would need more functionality from a 
physical embodiment than Section \ref{section:known-sep}'s programs
use in their timing-only mechanism.  We have preferred for the present
to investigate algorithms that use only the timing of tokens to 
overcome an unpredictable initial state.  

An obvious direction for future research is to move beyond rings
to other topologies.  We think it likely that some of the 
desiderata \desd{1}--\desd{5} will be relaxed for other 
topologies.  Figure \ref{figure:virtring} suggests how a virtual
ring, induced by a walk that includes all nodes, might be mapped 
upon a network.  Our protocols could be adapted to run on the 
virtual ring, and this might provide separated token circulation.  
Note that distance between tokens in the virtual ring could map 
to smaller distance in the underlying topology, because a node may 
appear more than once in the virtual ring.   The existence of 
a walk for which $m$ tokens can be separated by distance $d$ 
in the underlying topology is an open question.   Instead of 
mapping a complete walk of the network nodes, another strategy 
could be to map distinct rings upon a network so that they cover 
all nodes, and then hope to coordinate the timing of token 
circulation in these rings where they intersect.  

\begin{figure}[ht]
\quad
\begin{minipage}{0.45\columnwidth}
{\footnotesize
A complete walk of the network shown induces a virtual ring;
nodes in the center row of the network occur more than once
in the walk, conflicting with \desd{4}.  Shown are three tokens 
separated by distance at least two;  as all tokens synchronously
follow the walk, the separation by $d=2$ persists in the 
underlying network.  
}
\end{minipage}
\quad 
\framebox{
\begin{minipage}{0.45\columnwidth}
\quad\epsfig{file=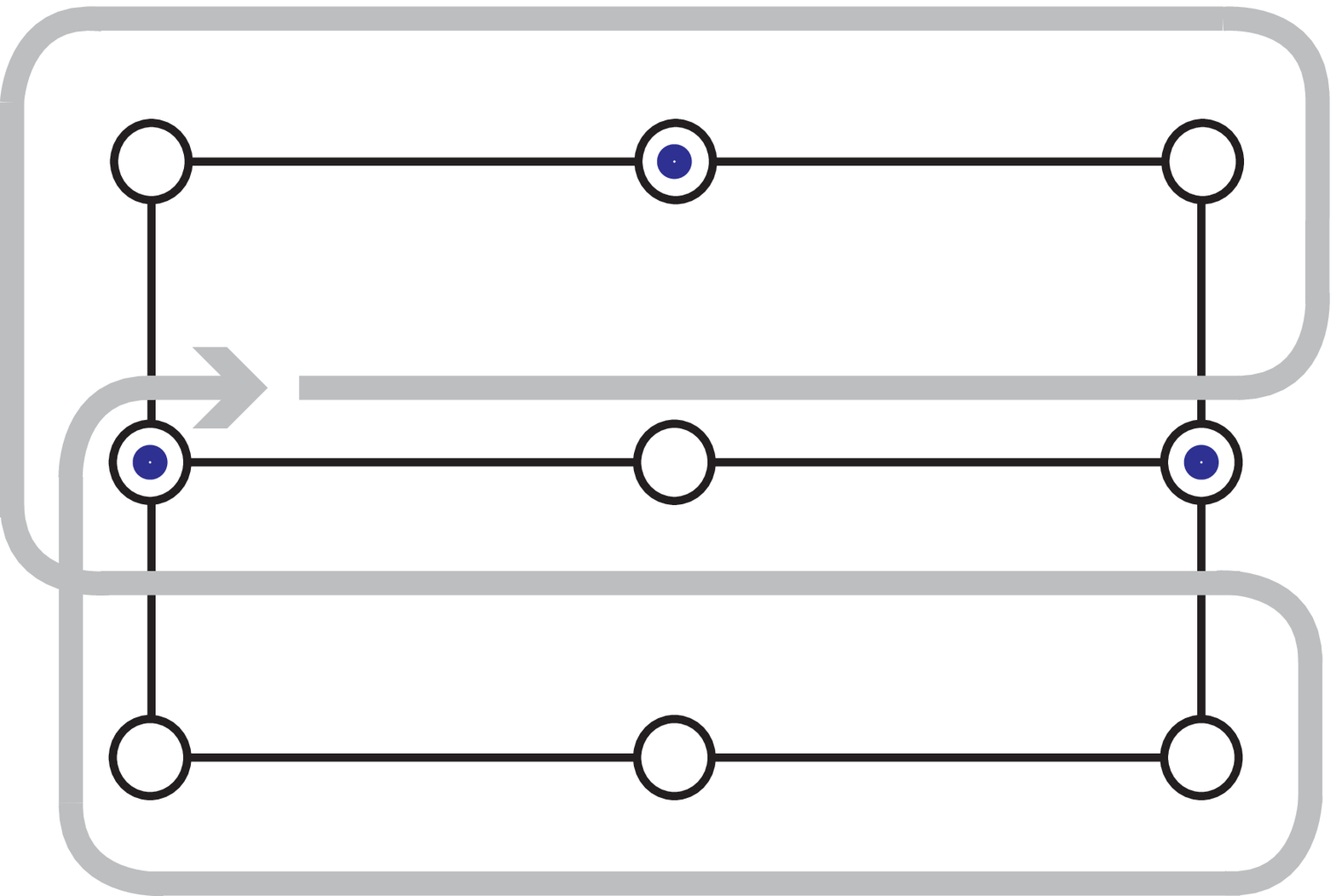,height=4cm,width=5.8cm}
\end{minipage}
}
\caption{Virtual Ring}
\label{figure:virtring}
\end{figure}

\end{document}